\documentclass[11pt]{amsart}

\usepackage{amsmath,amssymb,amsthm,amsaddr}
\usepackage{graphicx}
\usepackage{url}
\usepackage{enumitem}
\usepackage{bm}
\usepackage[textsize=small]{todonotes}
\usepackage{subfigure}
\usepackage{ifthen}
\usepackage[all]{nowidow}
\usepackage[margin=1in,includehead]{geometry}
\usepackage{thmtools}
\usepackage{thm-restate}

\newcommand{\casefmt}[1]{\vspace{1mm} \noindent \underline{#1}}

\newcommand{\newparentheses}[3]{%
  \expandafter\newcommand\csname #1\endcsname[1]{#2##1#3}%
  \expandafter\newcommand\csname #1L\endcsname[1]{\bigl#2##1\bigr#3}%
  \expandafter\newcommand\csname #1XL\endcsname[1]{\Bigl#2##1\Bigr#3}%
  \expandafter\newcommand\csname #1V\endcsname[1]{\left#2##1\right#3}}

\newparentheses{parens}{(}{)}
\newparentheses{braces}{\{}{\}}
\newparentheses{brackets}{[}{]}
\newparentheses{floor}{\lfloor}{\rfloor}
\newparentheses{ceil}{\lceil}{\rceil}
\newparentheses{abs}{|}{|}
\newparentheses{set}{\{}{\}}
\newparentheses{size}{|}{|}

\makeatletter
\newcommand{\onenewattribute}[3]{%
  \@ifundefined{#1}{\let\@@def\newcommand}{\let\@@def\renewcommand}%
  \expandafter\@@def\csname #1\endcsname[2][]{%
    \ifthenelse{\equal{##1}{}}%
    {#2\csname #3\endcsname{##2}}%
    {#2_{##1}\csname #3\endcsname{##2}}}}
\newcommand{\newattribute}[2]{%
  \onenewattribute{#1}{#2}{parens}%
  \onenewattribute{#1L}{#2}{parensL}%
  \onenewattribute{#1XL}{#2}{parensXL}%
  \onenewattribute{#1V}{#2}{parensV}}
\makeatother

\newattribute{OhOf}{\mathrm{O}}
\newattribute{ThetaOf}{\Theta}
\newattribute{OmegaOf}{\Omega}
\newattribute{ohOf}{\mathrm{o}}
\newattribute{omegaOf}{\omega}

\newattribute{alg}{\textsf{mAF}}

\newattribute{algstar}{\textsf{mAF*}}
\newattribute{replug}{\textsf{replug}}
\newattribute{replugdecorate}{\textsf{replug-decorate}}
\newattribute{replugdecoratetwo}{\textsf{replug-decorate*}}

\newattribute{dspr}{d_{\mathrm{SPR}}}
\newattribute{atbr}{a_{\mathrm{TBR}}}
\newattribute{dtbr}{d_{\mathrm{TBR}}}
\newattribute{dreplug}{d_{\mathrm{R}}}
\newattribute{ecut}{e}
\newattribute{mafsize}{m}
\newattribute{maf}{\textrm{MAF}}

\newcommand{\subtree}[2][]{%
  \ifthenelse{\equal{#1}{}}%
  {T(#2)}%
  {#1(#2)}}

\newcommand{\induced}[2][]{%
  \ifthenelse{\equal{#1}{}}%
  {T|#2}%
  {#1|#2}}

\newcommand{\notarxiv}[1]{}
\newcommand{\eat}[1]{}

\newtheorem{theorem}{Theorem}
\newtheorem{definition}{Definition}
\newtheorem{obs}{Observation}

\declaretheorem[name=Observation,sibling=obs]{re-obs}
\declaretheorem[name=Lemma,sibling=lemma]{re-lemma}
\declaretheorem[name=Corollary,sibling=corollary]{re-cor}

\begin{document}

\title{Chain Reduction Preserves the Unrooted Subtree Prune-and-Regraft Distance}

\author{Chris Whidden and Frederick A. Matsen IV}

\address{Program in Computational Biology \\
Fred Hutchinson Cancer Research Center \\
Seattle, WA, USA  98109\\
Tel.: +1-206-667-1311\\
Fax: +1-206-667-2437\\
}

\email{\{cwhidden,matsen\}@fredhutch.org}

\thanks{
This work was funded by National Science Foundation award 1223057 and 1564137.
Chris Whidden is a Simons Foundation Fellow of the Life Sciences Research Foundation.
The research of Frederick Matsen was supported in part by a Faculty Scholar grant from the Howard Hughes Medical Institute and the Simons Foundation.
}


\begin{abstract}
	\normalsize
The subtree prune-and-regraft (SPR) distance metric is a fundamental way of comparing evolutionary trees.
It has wide-ranging applications, such as to study lateral genetic transfer, viral recombination, and Markov chain Monte Carlo phylogenetic inference.
Although the rooted version of SPR distance can be computed relatively efficiently between rooted trees using fixed-parameter-tractable algorithms, in the unrooted case previous algorithms are unable to compute distances larger than 7.
One important tool for efficient computation in the rooted case is called \emph{chain reduction}, which replaces an arbitrary chain of subtrees identical in both trees with a chain of three leaves.
Whether chain reduction preserves SPR distance in the unrooted case has remained an open question since it was conjectured in 2001 by Allen and Steel, and was presented as a challenge question at the 2007 Isaac Newton Institute for Mathematical Sciences program on phylogenetics.

In this paper we prove that chain reduction preserves the unrooted SPR distance.
We do so by introducing a structure called a socket agreement forest that restricts edge modification to predetermined socket vertices, permitting detailed analysis and modification of SPR move sequences.
This new chain reduction theorem reduces the unrooted distance problem to a linear size problem kernel, substantially improving on the previous best quadratic size kernel.
\end{abstract}


\maketitle

\thispagestyle{empty}
\clearpage
\pagebreak


\emergencystretch=1em

\section{Introduction}
\label{sec:intro}

\setcounter{page}{1}

Molecular phylogenetic methods reconstruct evolutionary trees (a.k.a\ phylogenies) from DNA or RNA data and are of fundamental importance to modern biology~\cite{hillis96}.
Phylogenetic inference has numerous applications including investigating organismal relationships (the ``tree of life''~\cite{koonin2015turbulent}), reconstructing virus evolution away from innate and adaptive immune defenses~\cite{castro2012evolution}, analyzing the immune system response to HIV~\cite{haynes2012b}, designing genetically-informed conservation measures~\cite{helmus2007phylogenetic}, and investigating the human microbiome~\cite{matsen2015phylogenetics}.
Although the molecular evolution assumptions may differ for these different settings, the core algorithmic challenges remain the same: reconstruct a tree graph representing the evolutionary history of a collection of evolving units, which are abstracted as a collection of \emph{taxa} where each \emph{taxon} is associated with a DNA, RNA, or amino acid sequence.

Phylogenetic study often requires an efficient means of comparing phylogenies in a meaningful way.
For example, different inference methods may construct different phylogenies and it is necessary to determine to what extent they differ and, perhaps more importantly, which specific features differ between the trees.
In addition, the evolutionary history of individual genes does not necessarily follow the overall history of a species due to \emph{reticulate} evolutionary processes: lateral genetic transfer, recombination, hybridization, and incomplete lineage sorting~\cite{galtier2008dealing}.
Comparisons of inferred histories of genes to each other, a reference tree, or a proposed species tree may be used to identify reticulate events~\cite{beiko2005highways,whidden2014supertrees} and distance measures between phylogenies may be used to optimize summary measures such as supertrees~\cite{pisani2007supertrees,Steel2008-pn,bansal2010robinson,whidden2014supertrees}.

Numerous distance measures have been proposed for comparing phylogenies.
The Robinson-Foulds distance~\cite{robinson81} is perhaps the most well known and can be calculated in linear time~\cite{day85}.
However, the Robinson-Foulds distance has no meaningful biological interpretation or relationship to reticulate evolution.
Typically, distance metrics are either easy to compute but share this lack of biological relation, such as the quartet distance~\cite{brodal2004computing} and geodesic distance~\cite{owen2011fast}, or are difficult to compute such as the hybridization number~\cite{baroni05} and maximum parsimony distance~\cite{Bruen2008-wx,kelk2014complexity,moulton2015parsimony}.

The subtree prune-and-regraft (SPR) distance is widely used due to its biological interpretability despite being difficult to compute~\cite{baroni05,beiko2006phylogenetic}.
SPR distance is the minimum number of lateral gene transfer events required to transform one tree into the other (Figure~\ref{fig:spr}); it provides a lower bound on the number of reticulation events required to reconcile two phylogenies.
As such, it has been used to model reticulate evolution~\cite{maddison97,nakhleh05}.
In addition, the SPR distance is a natural measure of distance when analyzing phylogenetic inference methods which typically apply SPR operations to find maximum likelihood trees~\cite{Price2010-fi,Stamatakis2006-yz} or estimate Bayesian posterior distributions with SPR-based Metropolis-Hastings random walks~\cite{Ronquist2012-hi,bouckaert2014beast}.
Similar trees can be easily identified using the SPR distance, as random pairs of $n$-leaf trees differ by an expected $n - \Theta(n^{2/3})$ SPR moves~\cite{atkins2015extremal}.
This difference approaches the maximum SPR distance of $n - 3 - \floorV{(\sqrt{n-2} -1)/2}$~\cite{Ding2011-bj} asymptotically.
The topology-based SPR distance is especially appropriate in this context as topology changes have been identified as the main limiting factor of such methods~\cite{lakner2008efficiency,hohna2012guided,whidden2015quantifying}.
Moreover, the SPR distance has close connections to network models of evolution~\cite{baroni05,bordewich07,nakhleh05}.

Although it has these advantages, the SPR distance between both rooted and unrooted trees is NP-hard to compute~\cite{bordewich2005computational,hickey2008spr}, limiting its utility.
Despite the NP-hardness of computing the SPR distance between rooted phylogenies, recent algorithms can rapidly compare trees with hundreds of leaves and SPR distances of 50 or more in fractions of a second~\cite{whidden2010fast,whidden2014supertrees}.
This has enabled use of the SPR distance for inferring phylogenetic supertrees and lateral genetic transfer~\cite{whidden2014supertrees}, comparing influenza phylogenies to assess reassortment~\cite{dudas2014reassortment}, and investigating mixing of Bayesian phylogenetic posteriors~\cite{whidden2015quantifying,whidden2015ricci}.

SPR distances can be computed efficiently in practice for rooted trees by computing a maximum agreement forest (MAF) of the trees~\cite{hein96,allen01}.
Due to this MAF framework, the development of efficient fixed-parameter and approximation algorithms for SPR distances between rooted trees has become an area of active research~\cite{beiko2006phylogenetic,wu2009practical,bonet2009efficiently,whidden2013fixed,shi2014improved,chen2013faster} (see ~\cite{whidden2013fixed} for a more complete history), including recent extensions to nonbinary trees~\cite{whidden2015multifurcating,chen2015parameterized}, and generalized MAFs of multiple trees~\cite{shi2014approximation}.

Reduction rules form a key step of these fixed parameter-algorithms, including the subtree reduction rule~\cite{allen01}, chain reduction rule~\cite{allen01}, and cluster reduction rule~\cite{linz2011cluster}.
The first two optimizations collectively reduce the size of the compared trees to a linear-size problem kernel with respect to their SPR distance, while the third partitions the trees into smaller subproblems that can be considered independently.
These optimizations greatly reduce the search space that must be considered during the MAF search.

Most phylogenetic inference packages today use reversible mutation models to infer unrooted trees, motivating SPR calculation for unrooted trees.
However, unrooted SPR (uSPR) currently lacks some of the reduction rules available in the rooted case, making it difficult to analyze properties of the distance.
In particular, although the subtree reduction rule is applicable to the unrooted case, only a weaker version of the chain reduction rule, the $9k$-chain reduction rule~\cite{bonet2010complexity} has been proven to preserve the uSPR distance.
Correspondingly, a quadratic size kernel is the current state of the art for computing the uSPR distance, in contrast to the linear-size kernel for rooted SPR.
This kernel is sufficient for proving that the problem is fixed-parameter tractable, but does not make for a practical foundation on which to develop an efficient algorithm.
In part because of the lack of such a reduction rule, the best previous algorithm for computing the SPR distance between unrooted trees, due to Hickey et al.~\cite{hickey2008spr}, cannot compute distances larger than 7 or reliably compare trees with more than 30 leaves.

In this paper, we substantially advance understanding of and computational algorithms for the unrooted SPR distance.
We first introduce a new concept of socket agreement forests (SAFs) that allow us to look at uSPR paths in a new way and define equivalences between these paths,
Then, building on previous work by Hickey et al.\ \cite{hickey2008spr} and Bonet and St. John \cite{bonet2010complexity}, we prove the 2001 conjecture that the chain reduction preserves the uSPR distance, reducing the uSPR distance problem to a linear size problem kernel.

\section{Preliminaries}
\label{sec:prelim}

\begin{figure}[t]
	\subfigure[\label{fig:x-tree}]{\includegraphics[scale=1.2]{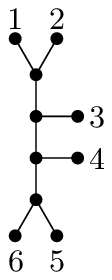}}
	\hspace*{\stretch{1}}
	\subfigure[\label{fig:subtree}]{\includegraphics[scale=1.2]{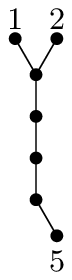}}
	\hspace*{\stretch{1}}
	\subfigure[\label{fig:induced}]{\includegraphics[scale=1.2]{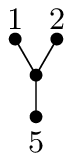}}
	\hspace*{\stretch{2}}
	\subfigure[\label{fig:spr}]{\includegraphics[scale=1.2]{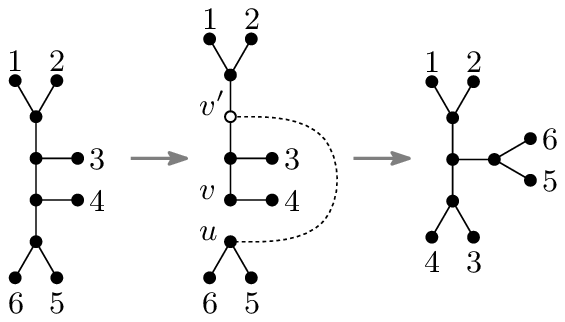}}

	\caption{(a) An unrooted $X$-tree $T$.
		(b) $T(V)$, where $V = \set{1,2,5}$.
		(c) $T|V$.
		(d) An SPR operation transforms $T$ into a new tree by \emph{pruning a subtree} and \emph{regrafting} it in another location.
	}
	\label{fig:trees}
\end{figure}

Nodes (i.e. vertices) of a tree graph with one neighbor are called \emph{leaves} and nodes with three neighbors are called \emph{internal nodes}.
An (unrooted binary phylogenetic) \emph{$X$-tree} is a tree $T$ whose nodes each have one or three neighbors, and such that the leaves of $T$ are bijectively labeled with the members of a label set $X$.
$T(V)$ is the unique subtree of $T$ with the fewest nodes that connects all nodes in $V$.
\emph{Suppressing} a node $v$ deletes $v$ and its incident edges; if $v$ is of degree 2 with neighbors $u$ and $w$, $u$ and $w$ are reconnected using a new edge $(u,w)$
The $V$-tree \emph{induced} by $T$ is the unique smallest tree $T|V$ that can be obtained from $T(V)$ by suppressing unlabeled nodes with fewer than three neighbours.
See Figure~\ref{fig:trees}.


An \emph{unrooted $X$-forest} $F$ is a collection of (not necessarily binary) trees $T_1, T_2, \ldots,\allowbreak T_k$ with respective label sets $X_1, X_2, \ldots, X_k$ such that $X_i$ and $X_j$ are disjoint, for all $1 \le i \ne j \le k$, and $X = X_1 \cup X_2 \cup \ldots \cup X_k$.
We say $F$ \emph{yields} the forest with components $T_1|X_1$, $T_2|X_2$, \ldots, $T_k|X_k$, in other words, this forest is the smallest forest that can be obtained from $F$ by suppressing unlabeled nodes with less than three neighbours.
For an edge set $E$, $F-E$ denotes the forest obtained by deleting the edges in $E$ from $F$ and $F \div E$ the yielded forest.
We say that $F \div E$ is \emph{a forest of $F$}.

A \emph{subtree-prune-regraft} (SPR) operation on an unrooted $X$-tree $T$ cuts an edge $e = (u,v)$.
This divides $T$ into subtrees $T_u$ and $T_v$, containing $u$ and $v$, respectively.
Then it introduces a new node $v'$ into $T_v$ by subdividing an edge of $T_v$, and adds an edge $(u,v')$.
Finally, $v$ is suppressed (Figure~\ref{fig:spr}).
In the following we assume all trees are unrooted unless otherwise stated.


We often consider a sequence of operations applied to a tree $T_1$ that result in a tree $T_2$.
These operations can be thought of as ``moving'' between trees and are also referred to as \emph{moves} (e.g. an SPR move).
A sequence of moves $M = m_1, m_2, \ldots, m_d$ applied to $T_1$ result in the sequence of trees $T_1 = t_0, t_1, t_2, \ldots, t_d = T_2$.
We call such sequences of trees a \emph{path} (e.g. an SPR path).

When considering how the tree changes throughout such sequences, it is often helpful to consider how nodes and edges of the tree change.
Formally, we construct a mapping $\varphi_{i,j}$ that maps nodes and edges of $t_i$ to $t_j$.
Each mapping $\varphi_{i,i+1}$ between adjacent trees is constructed according to the corresponding move $m_{i+1}$: nodes and edges of $t_i$ that are not modified by $m_{i+1}$ are mapped to the corresponding nodes and edges of $t_{i+1}$.
The deleted edge $(u,v)$ of $t_i$ is mapped to the newly introduced edge of $t_{i+1}$ (e.g. $(u,v')$ for an SPR move).
Deleted nodes are mapped to $\emptyset$.
Forward mappings $\varphi_{i,j}$, $i < j$, are constructed transitively.
Reverse mappings $\varphi_{j,i}$, $i < j$, are constructed analogously by considering the moves that construct the reverse sequence $t_d, t_{d-1}, \ldots, t_0$.

We will use these mappings implicitly to talk about how a tree changes throughout a sequence of moves.
With these mappings we can consider SPR tree moves as changing the endpoints of edges rather than deleting one edge and introducing another.
We say that an edge is \emph{broken} if one of its endpoints is moved by a rearrangement operation.

SPR operations give rise to a distance measure $\dspr{\cdot,\cdot}$ between $X$-trees, defined as the minimum number of SPR operations required to transform one tree into the other.
The trees in Figure~\ref{fig:three-spr}, for example, have SPR distance $\dspr{T_1,T_2} = 3$.
We will refer to a minimum-length path of SPR moves between two trees as an optimal SPR path.



Given trees $T_1$ and $T_2$ and forests $F_1$ of $T_1$ and $F_2$ of $T_2$, a forest $F$ is an \emph{agreement forest} (AF) of $F_1$ and $F_2$ if it is a forest of both forests.
$F$ is a \emph{maximum agreement forest} (MAF) if it has the smallest possible number of components.
We denote this number of components by $m(F_1, F_2)$.
For two unrooted trees $T_1$ and $T_2$, Allen and Steel~\cite{allen01} showed that the TBR distance is $m(T_1, T_2) - 1$.
Figure~\ref{fig:maf} shows an MAF of the trees in Figure~\ref{fig:three-spr}.

\section{Socket Agreement Forests}
\label{sec:sockets}

In this section we introduce a new type of agreement forest, socket agreement forests (SAFs).
SPR operations on general trees are difficult to analyze because they remove and introduce internal nodes.
SAFs solve this difficulty by including a finite set of predetermined \emph{sockets} which are the only nodes that can be involved in SPR operations and are never deleted or introduced.
Due to this fixed nature, SAFs are unsuitable for enumeration or determining a distance metric directly.
Instead, SAFs allow us to identify properties of independence with respect to optimal SPR paths and determine cases where such an optimal path can be modified to obtain a different optimal SPR path.
Proofs of the Observations, Corollary, and Lemma in this section can be found in the Appendix.
We use these concepts in Section~\ref{sec:chain-reduction} to prove that the chain reduction preserves the SPR distance between unrooted trees.

Define a \emph{socket forest} to be a collection of unrooted trees with special nodes, called \emph{sockets}.
Socket forests have special edges called \emph{connections} that must be between two sockets.
A collection of them is a \emph{connection set}.
Connections are not allowed to connect a socket to itself, although multiple connections to the same socket are allowed.

We define the \emph{underlying} forest of a socket forest $F$ to be the forest $F^*$ obtained from $F$ by deleting all connections and suppressing all unconnected sockets.
We will say that a socket forest $F$ \emph{permits an unrooted tree} $T$ if it is possible to add edges between the sockets of $F$, resolve any multifurcations in some way, and suppress unconnected sockets to obtain $T$.
Moreover, we say that a socket forest $F$ \emph{permits an SPR path} if each intermediate tree along the path is permitted by $F$.
Given two trees $T_1$ and $T_2$, a \emph{socket agreement forest} (SAF) is a socket forest that permits both $T_1$ and $T_2$.
Note that the underlying forest of an SAF is an AF of $T_1$ and $T_2$.

Because we will need to be precise concerning ways that connections are changing in socket forests, we offer the following clarifications.
First, because each socket is separately identified (e.g.\ with a numbering), any connection can be described irrespective of the other connections in a socket forest.
As with SPR moves on general trees, we consider the deletion and insertion of a connection as simply changing the endpoint of the connection.
As such, the ``new'' connection maintains the same identifier.
Thus, we can identify changes in a connection by the changes in the sockets it connects, again irrespective of the other connections in a socket forest.
The fact that connections are well defined irrespective of other connections implies that there is a well defined notion of equivalence of moves: two moves are equivalent if they both attach a given endpoint of the same connection to the same socket.

We will also use the following terminology.
A \emph{panel} is a component of a socket agreement forest.
A \emph{singleton panel} is a panel with one socket.
An \emph{SPR move} for a given connection set is the replacement of one connection in a connection set for another that does not introduce cycles.
We will denote SPR moves that replace a connection $c = (u,v)$ with a connection $(u,v')$ by $(u,v) \rightarrow (u,v')$ for short.
We say that this move \emph{breaks} the connection $c$.
Again, we can uniquely describe this move as changing the second endpoint of connection $c$ to socket $v'$, regardless of the current state of the socket forest.
An SPR move attaching to socket $v$ is \emph{terminal} for a given sequence of moves if subsequent moves maintain the connection endpoint attached to $v$.

Let $M = m_1, m_2, \ldots, m_k$ be an optimal sequence of SPR moves transforming tree $T_1$ into tree $T_2$ via a socket forest $F$.
We will often consider the sequence of trees $T_1, t_1, t_2, \ldots, t_k = T_2$ induced by these moves, that is the sequence of trees obtained by applying $M$ to a fully-connected socket forest configuration of $F$ that permits $T_1$ and results in $T_2$.
We say that two such trees are \emph{equivalent} if they are both permitted by the same binary phylogenetic tree.
In this way we can discuss sockets and panels in the trees, as shorthand for the sockets and panels in the socket forest configurations that correspond to each tree.

We say that two SPR moves $m_i$ and $m_j$, $i \neq j$ in such an optimal path are \emph{independent} if there exists another optimal sequence of SPR moves transforming $T_1$ into $T_2$ such that equivalent moves to $m_j$ and $m_i$ occur in a different order.
In contrast, connections to panels with multiple sockets may form cycles depending on the order of the moves.

\begin{restatable}{re-obs}{moveleaf}
	\label{obs:move-leaf}
	An SPR move that breaks an edge connected to a singleton panel is independent of any other SPR move in an optimal SPR path.
\end{restatable}

We next observe that modifying a terminal SPR move to use a different socket in the same panel of an underlying AF creates a new sequence of SPR moves that results in the same tree other than the modified connection.
Let $M = m_1, m_2, \ldots, m_k$ be a sequence of SPR moves transforming tree $T_1$ into tree $T_2$, $F$ a socket agreement forest that permits $M$, and $F^*$ the AF underlying $F$.
\begin{obs}
	\label{obs:modify-terminal}
	If $m_i = (u,v) \rightarrow (u,v')$ is a terminal move of $M$ and the component of $F^*$ containing $v'$ also contains a socket $v''$, then $M' = m_1, m_2, \ldots, m_{i-1}, m'_i,\allowbreak m_{i+1}, m_{i+2}, \ldots, m_k$ is a valid sequence of SPR moves, where $m'_i = (u,v) \rightarrow (u,v'')$.
\end{obs}

Note that a move $m_i$ is terminal with respect to the subsequence $m_i, m_{i+1}, \allowbreak \ldots, m_{j-1}$, where $m_j$ is the next move of the $v'$ endpoint moved by $m_i$.
Hence we can obtain a new sequence of SPR moves from $M$ that results in an equivalent tree by modifying the non terminal move $m_i$ to use a different socket in the same panel of the underlying AF.
We must also accordingly modify the subsequent move $m_j$.
Thus, we have the following corollary:

\begin{restatable}{re-cor}{modifynonterminal}
	\label{cor:modify-nonterminal}
	Suppose $m_i = (u,v) \rightarrow (u,v')$ is a non-terminal move of $M$ and the component $k$ of $F^*$ containing $v'$ also contains a socket $v''$.
	Let $m_j = (w,v') \rightarrow (w,x)$ be the next move in $M$ of the $v'$ endpoint moved by $m_i$.
	Then $M' = m_1, m_2, \ldots, m_{i-1}, m'_i, m_{i+1},\allowbreak m_{i+2}, \ldots, m_{j-1}, m'_j, m_{j+1}, m_{j+2}, \ldots, m_k$ is a valid sequence of SPR moves that results in an equivalent tree as $M$, where $m'_i = (u,v) \rightarrow (u,v'')$ and $m'_j = (w,v'') \rightarrow (w,x)$.
\end{restatable}

In other words, sockets of a given panel are interchangeable with respect to non terminal moves, and only the specific panel is important.

Given an AF $F'$ of two trees $T_1$ and $T_2$, we say that an SPR path between $T_1$ and $T_2$ is \emph{optimal with respect to $F'$} if there exists no shorter SPR path between $T_1$ and $T_2$ where each intermediate tree along the path is permitted by $F'$.
\begin{restatable}{re-lemma}{saf}
	\label{lem:saf}
  Let $F$ be a socket agreement forest of two trees $T_1$ and $T_2$.
	Then there exists an SPR path between $T_1$ and $T_2$ that is permitted by $F$ and optimal with respect to the AF $F^*$ underlying~$F$.
\end{restatable}

Thus, if we can construct a socket agreement forest for $T_1$ and $T_2$, we can be assured of a valid SPR path between $T_1$ and $T_2$ that is optimal with respect to the underlying agreement forest.
However, it is not trivial to compare socket agreement forests by the length of the SPR path between the trees, and thus they are only a partial analogue of maximum agreement forests.

\section{Unrooted chain reduction is distance-preserving}
\label{sec:chain-reduction}

\begin{figure}[t]
	\hspace*{\stretch{1}}
	\includegraphics[scale=1]{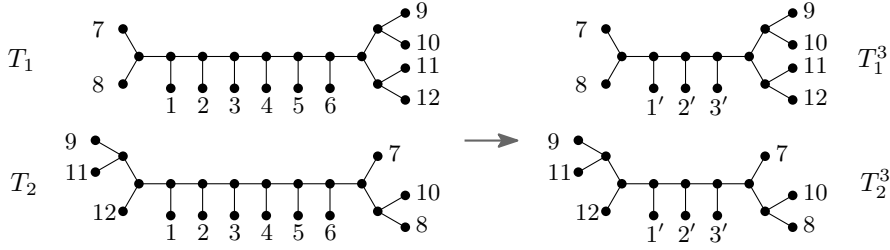}
	\hspace*{\stretch{1}}

	\caption{An application of the chain reduction rule transforming two trees $T_1$ and $T_2$ into trees $T_1^3$ and $T_2^3$ with the same SPR distance.
	}
	\label{fig:chain-reduction}
\end{figure}

In this section we investigate the chain reduction rule for unrooted trees, noting that chain reduction for rooted trees is a key component of fast rooted SPR algorithms and TBR algorithms.

\begin{definition}[Chain Reduction Rule]
Replace a chain of subtrees that occur identically in both trees with three new leaves with new labels oriented to preserve the direction of the chain (Figure~\ref{fig:chain-reduction}).
\end{definition}


\begin{definition}[Subtree Reduction Rule]
Replace a pendant subtree that occurs identically in both trees with a single new leaf~\cite{hickey2008spr}.
\end{definition}

In conjunction with the subtree reduction rule, the chain reduction rule for rooted trees reduces the number of leaves in each tree to a linear function of the MAF size $k$---at most $28k$---while preserving the rooted SPR and TBR distance.
This rule is thus a key element of fixed-parameter tractability proofs for rooted SPR~\cite{bordewich2005computational} and TBR~\cite{allen01}.
Allen and Steel~\cite{allen01} conjectured that the chain reduction rule also holds for unrooted SPR, but this claim has been difficult to prove or disprove.
Bonet and St. John~\cite{bonet2010complexity} proved that a relaxed version of the chain reduction rule holds for \underline{unrooted} trees: the $9k$-chain reduction rule, which replaces each chain of subtrees with $9k$ leaves rather than 3.
Although useful to prove the fixed-parameter tractability of uSPR, this ``reduction" will typically greatly inflate the size of the trees in practice.
The resulting quadratic bound on the number of leaves---at most $76k^2$---is most impractical for computing the uSPR distance.

Previous work on unrooted trees has identified four cases that must be considered to prove that the chain reduction rule holds for unrooted trees, depending on which chain edges are broken by an optimal sequence of uSPR moves.
Bonet and St. John~\cite{bonet2010complexity} proved that the first two cases preserve uSPR distance and the latter two do not reduce it by more than 1.
The basic idea behind their proofs (inspired by a similar idea of Hickey et al.~\cite{hickey2008spr}) was to alter the initial trees to obtain a new pair of trees, each of which differs from the original by one SPR, with the common chain as subtrees to which one can apply the subtree reduction rule.
This directly gives a lower bound of two less than the distance, as shown by Hickey et al.~\cite{hickey2008spr}.
Bonet and St. John refined this approach by adding two additional elements to the chain and removing two of the common chain subtrees.

In order to explain their procedure, we make the following trivial observation.
\begin{obs}
Assume two ordered pairs of trees $(T_1,T_2)$ and $(T_1',T_2')$, each tree in a pair with the same leaf set, such that there is a bijection between the two leaf sets that when applied to the ordered pair $(T_1,T_2)$ results in the ordered pair $(T_1',T_2')$.
Then $\dspr{T_1,T_2} = \dspr{T_1',T_2'}$.
\end{obs}
\noindent
We will say pairs of trees satisfying the hypothesis of this observation \emph{have the same tanglegram}~\cite{Matsen2016-cv,Billey2017-ci}, and so the observation can be rephrased as saying that pairs of trees with the same tanglegram have the same SPR distance.

\begin{figure}[t]
	\centering
	\hspace*{\stretch{1}}
	\subfigure[\label{fig:chain-notation}]{\includegraphics[scale=0.9]{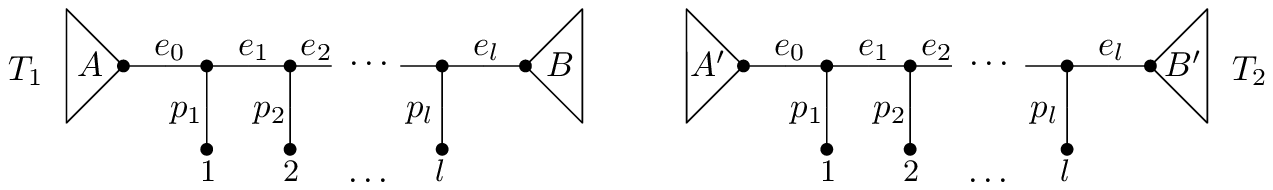}}
	\hspace*{\stretch{1}}

	\hspace*{\stretch{1}}
	\subfigure[\label{fig:chain-refinement}]{\includegraphics[scale=0.9]{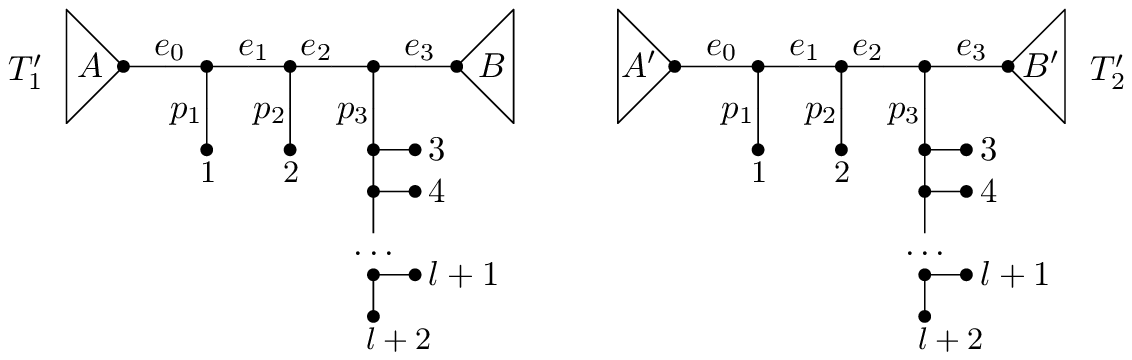}}
	\hspace*{\stretch{1}}

	\hspace*{\stretch{1}}
	\subfigure[\label{fig:chain-removal}]{\includegraphics[scale=0.9]{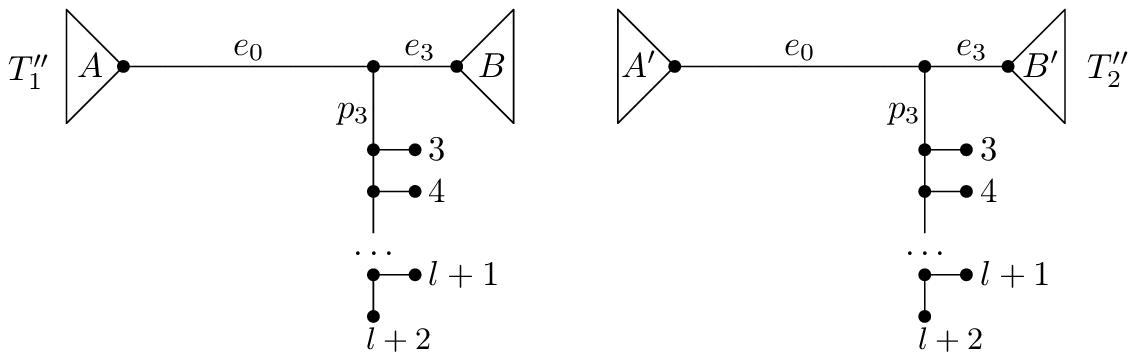}}
	\hspace*{\stretch{1}}

	\caption{(a)~Notation for a common chain in trees $T_1$ and $T_2$.
		(b)~The refined trees $T_1'$ and $T_2'$ after applying the refinement of Bonet and St. John but before removing any pendant edges.
		Two additional labels are added to the chain, which does not increase the SPR distance.
		These trees have a chain of length 3 after applying the subtree reduction rule.
		(c)~The refined trees $T_1''$ and $T_2''$ after removing the pendant edges $p_1$ and $p_2$.
		Each refined tree can be obtained by applying a single SPR operation to the original and mapping the labels $1, 2, \ldots l$ to $3, 4, \ldots l+2$.
	}
	\label{fig:chains}
\end{figure}

Here we refine these previous bounds and thus show that the chain reduction rule is distance-preserving: an application of the chain reduction rule does not change uSPR distance.
To discuss this formally, we introduce notation for common chains (Figure~\ref{fig:chain-notation}).
Pendant edges of a chain $1, 2, \ldots, l$ are labeled $p_1, p_2, \ldots, p_l$.
The edges connecting the chain are labeled $e_0, e_1, \ldots, e_l$.
$e_0$ is connected to subtree $A$ and $A'$ in $T_1$ and $T_2$, respectively.
Similarly, $e_l$ is connected to $B$ and~$B'$.

Bonet and St. John first add two pendant edges with leaves labelled $l + 1$ and $l + 2$ on the edge leading to $l$ for each tree, obtaining new trees $T_1^{l+2}$ and $T_2^{l+2}$.
Any sequence of SPR moves that transforms $T_1$ into $T_2$ also transforms $T_1^{l+2}$ into $T_2^{l+2}$.
Thus, $\dspr{T_1, T_2} = \dspr{T_1^{l+2}, T_2^{l+2}}$.

Next, they apply SPR to each of $T_1^{l+2}$ and $T_2^{l+2}$, attaching $B$ to $e_3$ and $B'$ to $e_3$ (Figure~\ref{fig:chain-refinement}), obtaining new trees $T_1'$ and $T_2'$.
$T_1'$ and $T_2'$ are at most 2 SPR moves closer than $T_1$ and $T_2$, that is, $\dspr{T_1', T_2'} \ge \dspr{T_1, T_2} - 2$, simply by virtue of being a pair of moves away from $T_1^{l+2}$ and $T_2^{l+2}$.
$T_1'$ and $T_2'$ have a chain of length 3 and, after applying the subtree reduction rule and relabeling leaves, have the same tanglegram as the trees obtained by chain reduction applied to leaves $1, \ldots, l$ of $T_1$ and $T_2$.
By the implied equality, chain reduction reduces the SPR distance by at most 2.

Finally, Bonet and St. John remove two of the three pendant edges (e.g. $p_1$ and $p_2$ in Figure~\ref{fig:chain-removal}) to obtain trees $T_1''$ and $T_2''$.
Let $\delta = \dspr{T_1',T_2'} - \dspr{T_1'',T_2''}$.
A leaf is moved at most once in an optimal SPR path (for a formal argument, see Observation~\ref{obs:move-leaf}).
Thus, $\delta=1$ if one of the removed pendant edges is moved as part of an optimal SPR path, and $\delta=2$ if both are moved.

Observe that a pair of trees with the same tanglegram as the pair $T_1'', T_2''$ can be obtained from $T_1, T_2$ by a single SPR move applied to each tree of the pair.
$T_1''$ can be obtained from $T_1$ by attaching $B$ to $e_0$ and changing the labels of leaves $1,2, \ldots l$ to $3, 4, \ldots l+2$.
$T_2''$ can be obtained from $T_2$ by attaching $B'$ to $e_0$ and applying the same relabelling.
This implies that $T_1''$ and $T_2''$ are at most 2 SPR moves closer than $T_1$ and $T_2$.
We thus have that $\dspr{T_1'', T_2''} \ge \dspr{T_1, T_2} - 2$.
Because $\dspr{T_1'', T_2''} = \dspr{T_1',T_2'} - \delta$, we can conclude that $\dspr{T_1', T_2'} \ge \dspr{T_1, T_2} + \delta - 2$.
An application of the chain reduction can never increase the SPR distance, because any move applied to the chain-reduced tree has an equivalent in the original trees, and equality of chain-reduced trees implies equality of the original trees.
Therefore, the chain reduction does not change the SPR distance in the case that both removed edges were involved in an optimal SPR path (i.e.\ the case of $\delta=2$), and the chain reduction reduces the distance by at most one when one removed edge was involved in such a path (i.e.\ the case of $\delta=1$).
Bonet and St. John showed that the case where $\delta=0$ can be transformed to a case where $\delta=1$, implying that the chain reduction reduces the distance by at most one in this case as well.

We look at this reduction in a different but related way.
Let $\bar T_i''$ be $T_i''$ after relabeling the leaves $3, 4, \ldots l+2$ to $1, 2, \ldots l$.
One can use the same moves as found for going from $T_1''$ to $T_2''$ in order to get from $\bar T_1''$ to $\bar T_2''$.
We thus have a path $T_1$ to $\bar T_1''$ to $\bar T_2''$ and, finally, to $T_2$.
This sequence does not break any edge in the common chain, as the common chain becomes a common subtree in $\bar T_1''$ and $\bar T_2''$.
We can thus directly apply this sequence to the chain-reduced trees.
If using the resulting sequence of SPR moves avoids two SPR operations (involving $p_1$ and $p_2$) then this strategy preserves the SPR distance.
On the other hand, this strategy requires one extra SPR move when maintaining the common chain only avoids one SPR operation.
By carefully analyzing the sequence of SPR moves, we show that only the $T_2$ transformation is necessary in such cases.
To do so, we generate a sequence of moves that proceeds from $T_1$ to $\bar T_2''$ to $T_2$, without touching $\bar T_1''$.
Avoiding $\bar T_1''$ avoids the extra SPR move and tightens the bound to equality.

\begin{figure}[t]
	\centering
	\hspace*{\stretch{1}}
	\includegraphics[scale=1.]{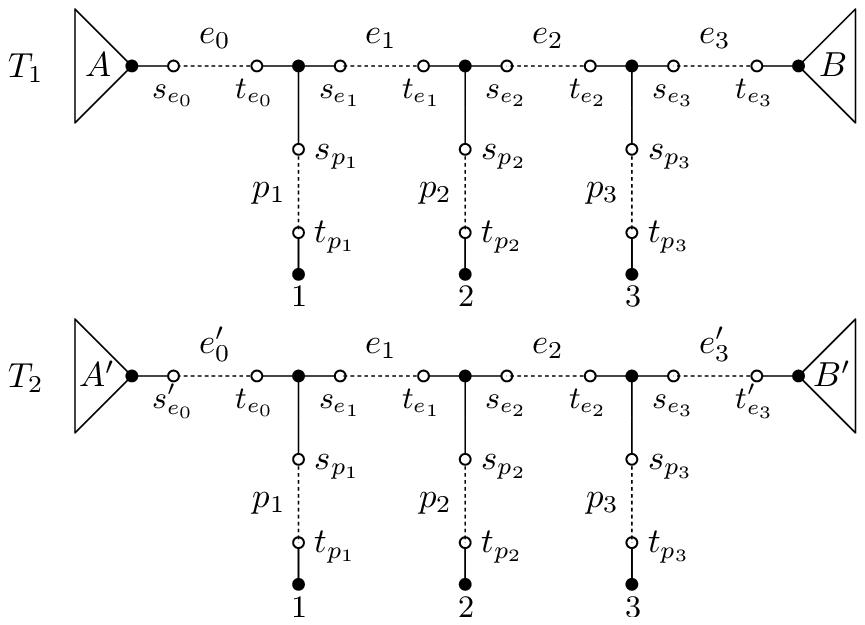}
	\hspace*{\stretch{1}}

	\caption{
		Configurations of the socket agreement forest $F'$ that permit $T_1^3$ and $T_2^3$. In the proof of Theorem~\ref{thm:chain-reduction} we transform a sequence $M$ of moves from $T_1^3$ to $T_2^3$ into a sequence $M'$ that maintains the chain and results in this $T_2$ configuration.
	}
\label{fig:chain-sockets}
\end{figure}

\begin{theorem}
\label{thm:chain-reduction}
The chain reduction rule does not change uSPR distance.
\end{theorem}
\begin{proof}
Let $T_1$ and $T_2$ be trees with a common chain $1, 2, \ldots, l$.
Let $T_1^3$ and $T_2^3$ be the result of applying the chain reduction rule to $T_1$ and $T_2$, labeled as in Figure~\ref{fig:chain-notation} with $l=3$.
Let $\dspr{T_1, T_2} = d$.
Let $M = m_1, m_2, \ldots, m_k$ be an optimal sequence of moves transforming $T_1^3$ into $T_2^3$.
As above, $d \geq k$.
Let $F$ be a socket agreement forest that permits $M$.
We consider four cases.

\casefmt{Case 1:} There exists an optimal sequence of moves $M$ transforming $T_1^3$ into $T_2^3$ without breaking $p_1$, $p_2$, $p_3$, $e_1$, or $e_2$.
Insert elements $4, 5, \ldots, l$ on $e_1$ to obtain trees isomorphic to $T_1$ and $T_2$.
$M$ maintains the common chain as an induced subgraph and, so, does not move the inserted elements.
Thus, $d = \dspr{T_1, T_2} = \dspr{T_1^3, T_2^3} = k$.

\casefmt{Case 2:} $M$ breaks two or three pendant edges from the set $\set{p_1, p_2, p_3}$.
The transformation of Bonet and St. John~\cite{bonet2010complexity} recalled above shows that $d = k$.

\casefmt{Case 3:} $M$ breaks exactly one pendant edge, $p_x$.
We will transform $M$ into a sequence of moves that does not break any pendant edges $p_i$ nor chain edges $e_1$ or $e_2$.
Then by Case 1, $d = k$.

By Observation~\ref{obs:move-leaf}, we can assume that $M$ moves $p_x$ last, so that $m_k$ breaks $p_x$ and $m_i$ does not, for all $1 \le i < k$.
We modify $M$ to $M' = m'_1, m'_2, \ldots, m'_{k-1}$ to avoid moves that change the middle of the chain (edges $p_2$, $e_1$, and $e_2$) but result in an equivalent final tree.
These moves will be redirected to the ends of the chain at $e_0$ and $e_3$.
We set up this redirection so that at least one of $A'$ and $B'$ will be attached to the correct end of the chain.
As such, we obtain a socket agreement forest $F'$ that permits $M'$ by making $e_0$ and $e_3$ into connections if they were not already, which results in at most 4 more sockets (see Figure~\ref{fig:chain-sockets}, which shows possible sockets along each edge).
We refer to the sockets of a connection $c$ by $s_c$ and $t_c$ (e.g. $s_{e_0}$ and $t_{e_0}$).
The connections corresponding to $e_0$ in $T_1$ and $T_2$ may differ, and as such we label them differently, as $e_0$ and $e'_0$, respectively.
Similarly, $e_3$ and $e'_3$ denote the connections that correspond to $e_3$ in $T_1$ and $T_2$, respectively.

Specifically, we apply the following move modification rules:
\begin{enumerate}[label=\Roman*.,topsep=5pt]
	\item If a move $m_i$ attaches a connection to a socket of $e_1$, $p_2$, or $e_2$, we instead redirect it to one end or another of the chain to define $m_i'$: if $x=1$, attach it to $t_{e_0}$; if instead $x=2$ or $3$, attach it to $s_{e_3}$. Recall that multiple connections may connect to the same socket, so this does not prevent any future moves, including other redirected moves.
		A subsequent move $m_j$ of that connection does so from its new position, defining the equivalent move $m_j'$.
	\item If a move $m_i$ attaches a connection to a socket of $p_1$, $m'_i$ instead attaches it to $t_{e_0}$. Similarly, connections to $p_3$ are redirected to $s_{e_3}$.
	\item
    We also redirect breaks in the middle of the chain to one side or the other: if a move $m_i$ breaks $e_1$ or $e_2$, define $m_i'$ to be the corresponding move that breaks $e_0$ or $e_3$, respectively, and attaches to the same side of that connection. For example, if $m_i$ changes $e_1 = (s_{e_1},t_{e_1})$ to $(s_{e_1},x)$ then $m'_i$ changes $e_0 = (s_{e_0},t_{e_0})$ to $(s_{e_0},x)$.
    In this case where $m_i$ breaks $e_1$, any subsequent move $m_j$ that would attach a connection to the $e_0$--$p_1$ path (i.e. sockets $t_{e_0}$, $s_{e_1}$, $s_{p_1}$, and $t_{p_1}$) before leaf $1$ returns to the chain instead attach to $s_{e_0}$.
    In the case that $m_i$ breaks $e_2$, we redirect connections to the $e_2$--$p_2$ path to $t_{e_3}$ in the analogous manner.
	Note that (thinking for a moment in terms of trees rather than socket forests) an SPR move breaking $e_1$ contracts the $e_0$--$p_1$ path into a single edge adjacent to leaf $1$ in the original tree.
		Moreover, no move of $M$ except $m_k$ breaks such an edge, and $M'$ does not apply $m_k$.
    This implies that no subsequent moves in $M'$ try to break $e_0$ (resp.\ $e_3$), after they have already been broken by transforming an $e_1$ break into an $e_0$ break.
	\item Any move $m_i$ not covered by the previous rules is replaced by an equivalent move $m'_i$ attaching the same end of the same numbered connection to the same socket.
\end{enumerate}

By Observation~\ref{obs:modify-terminal} and Corollary~\ref{cor:modify-nonterminal}, attaching to $t_{e_0}$ or $s_{e_3}$ in place of $e_1$, $e_2$, $p_1$, $p_2$, or $p_3$ sockets does not prevent any allowed moves if we maintain the chain as a single panel.
Similarly, we do not break any of the $p_i$s until $m_k$, so breaking $e_0$ in place of $e_1$ (or $e_3$ in place of $e_2$) does not prevent any moves.
Thus, $M'$ transforms $T_1^3$ into a tree consisting of $A'$ and $B'$ attached in some manner to the chain $1, 2, 3$.

Now, consider the location of $A'$ and $B'$.
As defined above, socket $s'_{e_0}$ is the socket of a panel $A''$ of $A'$ that is connected to the chain in $T_2$ by a connection $e'_0$.
Similarly, let socket $t'_{e_3}$ be the socket of a panel $B''$ of $B'$ that is connected to the chain in $T_2$ by a connection $e'_3$.
We will first show that one of these is in the correct location, that is, $A'$ is connected to the left side of the chain ($e'_0$ connects $s'_{e_0}$ of $A''$ and $t_{e_0}$)  or $B'$ is connected to the right side (i.e. $e'_3$ connects $t'_{e_3}$ of $B''$ and $s_{e_3}$).

	There are three possible events that joined $A''$ to the chain in the original sequence of tree moves, which we describe first in socket forest terms and then parenthetically with respect to ``classical'' SPR moves on trees without sockets:
\begin{enumerate}[label=\roman*.,topsep=5pt]
	\item $e'_0$ was attached to $t_{e_0}$, or $e'_0$ was attached to a socket of $p_1$, $e_1$, $p_2$, or $e_2$ and then $m_d$ moved leaf 1 to recreate the chain ($A''$ was moved to the chain).
	\item $e'_0$ was attached to $s'_{e_0}$ (the chain was moved to $A''$).
	\item $e'_0$ was never changed (a subtree on the path from $A''$ to the chain was moved).
\end{enumerate}
Note that it is possible that two or more such events occured to $e'_0$ during the application of $M$.
We consider only the last such event.
Similarly, there are three analogous events that can join $B'$ to the chain, depending on whether $e'_3$ was attached to the chain, $t'_{e_3}$, or never changed. We consider each pair of events with respect to our modified $M'$.

If $A''$ was moved to the chain by $M$, then $e'_0$ was either attached to $t_{e_0}$ by $M$ or attached to a socket in \set{$p_1$, $e_1$, $p_2$, $e_2$} and $x=1$.
Our first modification rule therefore implies that $M'$ attaches $e'_0$ to $t_{e_0}$.
Similarly, if $B''$ was moved to the chain by $M$ then $M'$ attaches $e'_3$ to $s_{e_3}$.

If the chain was moved to both $A''$ and $B''$ then either $e_0 = e'_0$, $e_3 = e'_3$, or there must be two moves $m_i$, $m_j$, so that $m_i$ attaches w.l.o.g. $e'_0$ to the chain and $m_j$ connects $e'_0$ to $s'_{e_0}$.
In the first two cases, the fact that moving $p_x$ with $m_k$ results in $T_2^3$ implies that at least one of $A''$, $B''$ is in the correct location after applying $M'$.
In the latter case, $m_i$ either attached $e'_0$ to $t_{e_0}$ or $m_i$ attached $t_{e_0}$ to a socket of $p_1$, $e_1$, $p_2$, or $e_2$ and then $m_k$ moved leaf 1, that is $x=1$.
Our rules again imply that $m'_i$ attaches $e'_0$ to $t_{e_0}$.
Then $m'_j$ attaches $e'_0$ to $s'_{e_0}$ by Observation~\ref{obs:modify-terminal} and Corollary~\ref{cor:modify-nonterminal}.

If both $e'_0$ and $e'_3$ never change during $M$, i.e.\ a subtree on the path from $A''$ to the chain was moved, and similarly for $B''$, then the fact that moving $p_x$ with $m_k$ results in $T_2^3$ again implies that at least one of $A''$, $B''$ is in the correct location.

Thus we conclude that one of the subtrees (say $A'$) must be in the correct location.
The fact that $M$ results in $T_2^3$ implies that $e'_3$ exists but connects $t'_{e_3}$ to a socket other than $s_{e_3}$ after applying $M'$ (in fact, our rules imply that it is connected to $t_{e_0}$).
We apply a final move $m'_k$ to connect $e'_3$ to $s_{e_3}$, replacing the $p_x$ move of $M$ in a one-sided variant of Bonet and St. John's refinement.
We then have that $m'_1, m'_2, \ldots, m'_{k-1}, m'_k$ transforms $T_1^3$ into $T_2^3$ while maintaining the common chain.

\casefmt{Case 4:} None of the above. $M$ breaks at least one of $\set{e_1, e_2}$ and does not break $p_1$, $p_2$, or $p_3$.

We first observe that breaking $e_1$ and $e_2$ also effectively breaks $p_2$, so Case 3 applies.
Now, suppose $M$ breaks exactly one of $e_1$ or $e_2$ and refer to the connection as $e_y$.
Modify $M$ to $M' = m'_1, m'_2, \ldots, m'_k$ using the modification rules I through IV, with $y$ in place of $x$ as the criteria for deciding the side of the chain to which to redirect connections.
We again consider how the chain becomes connected to $A''$ and $B''$.
The above argument holds if we move both $A''$ and $B''$.
Moreover, w.l.o.g. if we do not move $A''$ then we either connect the chain directly to $A''$ or break pendant edges along the path from $A''$ to the chain.
In either event, the fact that we do not break any edge $p_i$ in $M$ along with our rules for $M'$ imply that $e'_0$ connects $s'_{e_0}$ and $t_{e_0}$ and that $e'_3$ connects $s_{e_3}$ and $t'_{e_3}$ after applying $m'_k$.
Thus, $d = k$ and the claim follows.
\end{proof}

Combining the subtree and chain reduction rules we achieve a reduction procedure for computing the SPR distance between unrooted trees that results in a linear-size kernel.
This combined reduction procedure provides the base step for efficient SPR distance calculations on unrooted trees.

\begin{restatable}{re-cor}{reducetrees}
	\label{cor:reduce-trees}
	Let $T_1$ and $T_2$ be a pair of unrooted trees.
	Repeatedly applying the subtree and chain reduction rules to $T_1$ and $T_2$ until neither rule is applicable results in a pair of trees $T_1'$ and $T_2'$ with at most $28\dspr{T_1, T_2}$ leaves and such that $\dspr{T_1, T_2} = \dspr{T_1', T_2'}$.
\end{restatable}

\section{Conclusions}

We have worked to extend understanding of and methods to calculate the SPR distance between unrooted trees in several directions.
The maximum agreement forest framework used to prove properties of the rooted SPR distance and analyze algorithms for computing the distance can not be directly applied to the unrooted case.
Instead, we developed a more general representation called a socket agreement forest.
SAFs cannot determine the unrooted SPR distance, and instead are useful for determining notions of independent SPR moves and equivalences between SPR rearrangement sequences.
We used these ideas to prove the long-standing conjecture that the chain reduction rule preserves the SPR distance between unrooted trees.
Repeatedly interleaving the chain reduction rule and subtree reduction rule provides the first pillar of an efficient fixed-parameter algorithm---reducing the problem to a linear-size problem kernel.
This is a major improvement over the previous best quadratic-size problem kernel for this problem.
Chain reduction is a key subroutine of our new \textsf{uspr} software~\cite{uspr} for computing the unrooted SPR distance which can quickly compute distances up to 14 between trees with 50 leaves (manuscript in preparation).
Moreover, it is likely that our SAF framework will lead to new insights and algorithms for computing the unrooted SPR distance and related phylogenetic distances just as the MAF framework did for the rooted SPR distance.

\pagebreak

\appendix

\section*{Selected Proofs and Figures}

\begin{figure}[h]
	\hspace*{\stretch{1}}
	\subfigure[\label{fig:three-spr}]{\includegraphics[scale=1.2]{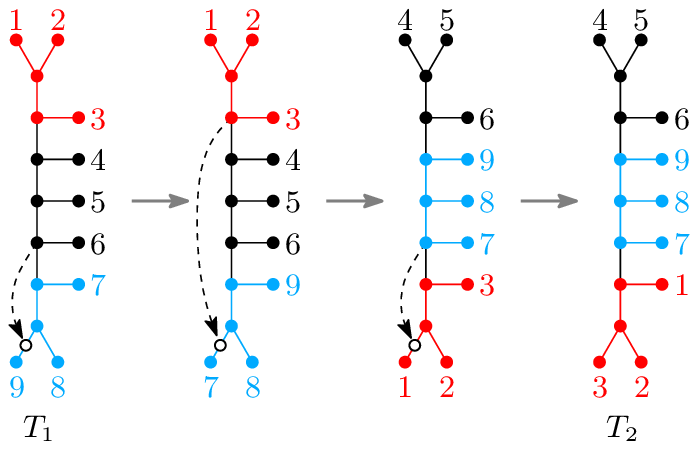}}
	\hspace*{\stretch{2}}
	\subfigure[\label{fig:maf}]{\includegraphics[scale=1.2]{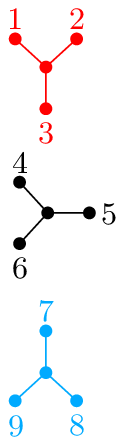}}
	\hspace*{\stretch{1}}

	\caption{(a) Three unrooted SPR operations transform a tree $T_1$ into another tree $T_2$.
		(b) A MAF of $T_1$ and $T_2$.
	}
	\label{fig:spr-distance}
\end{figure}

\moveleaf*
\begin{proof}
Consider an SPR move $m_i$ that breaks a connection $(u,v)$ and creates a connection $(u,v')$.
We observe that an SPR move must leave the singleton panel connected, and thus $u$ must be the socket in the singleton panel.
Performing any such move will not change acyclicity when done in any context, thus it does not prevent any subsequent moves.
This implies that an optimal SPR path does not break the singleton panel connection more than once, as one could simply remove the first such SPR move to obtain a shorter SPR path resulting in an equivalent tree.
	Then the sequences $m_i, m_1, m_2, \ldots, m_{i-1}, m_{i+1}, m_{i+2}, \ldots, m_k$ and $m_1, m_2, \ldots, m_{i-1}, m_{i+1}, m_{i+2},\allowbreak \ldots, m_k, m_i$ both result in $T_2$.
	Therefore $m_i$ is independent from each other move $m_j$ in the sequence.
\end{proof}

\modifynonterminal*
\begin{proof}
	Suppose, for the purpose of obtaining a contradiction, that the corollary is false: either some move $m'_q$ is not a valid SPR move or the sequence $M'$ does not result in an equivalent tree as the sequence $M$.

First, suppose that some move $m'_q$ is not a valid SPR move, i.e.\ suppose that some intermediate connection set does not form a tree in the sequence of socket agreement forests $t'_1, t'_2, \ldots, t'_k$ induced by applying $M'$ to $T_1$.
Let $q$ be the smallest corresponding index.
Then $m'_q$ cuts an edge $(x,y)$ between subtrees $T^x$ and $T^y$ of $t'_{q-1}$ and attempts to add an edge $(x, y')$ such that $y' \in T^x$.
Observe that $t_{q-1}$ and $t'_{q-1}$ differ only in the use of socket $v'$.
This implies that the path from $x$ to $y'$ in $t'_{q-1}$ contains an edge $(z, v')$.
However, $s'$ is in the same component of $F^*$ as $v$.
Then, $(z,v')$ and $(z,v)$ connect the same components of $F^*$ and, thus, the fact that move $m_q$ is valid implies that $m'_q$ is a valid SPR move.

Second, suppose that $t_k$ is not equivalent to $t'_k$.
We observe that $t'_{j-1}$ differs from $t_{j-1}$ only in the use of socket $v'$.
We have already shown that $M'$ is a valid sequence of SPR moves and we have that each $m'_r$ is equivalent to $m_r$, for all $j < r \le k$.
Thus, the fact that $m'_j$ is equivalent to $m_j$ implies that $t'_j$ is equivalent to $t_j$ and, moreover, each $t'_r$ is equivalent to $t_r$, for all $j < r \le k$, contradicting the supposition.
\end{proof}

\saf*
\begin{proof}
Let $F$ be a socket agreement forest of two trees $T_1$ and $T_2$.
Let $F^*$ be the AF underlying $F$.
Let $d$ be the length of an optimal SPR path from $T_1$ to $T_2$ with respect to $F^*$; such a path must exist as $F^*$ is an AF of $T_1$ and $T_2$.
To show the lemma, we will prove by induction on $d$ that $F$ permits an SPR path between $T_1$ and $T_2$ that is optimal with respect to $F^*$.

For the base case, suppose $d=1$.
Then $T_1$ and $T_2$ differ by a single SPR move, such that our path $P = T_1, T_2$.
The fact that $F$ is a socket agreement forest of $T_1$ and $T_2$ implies that $P$ is permitted.
This path is optimal with respect to $F^*$.

Now, suppose the claim holds for all $d^\circ < d$.
Let $t_1$ be a tree that is adjacent to $T_1$ on an optimal SPR path between $T_1$ and $T_2$ with respect to $F^*$.
If $F$ permits $t_1$, then, by the inductive hypothesis, there is an SPR path $P = t_1, t_2, \ldots, t_d$ (where $t_d = T_2$) such that $F$ permits $P$.
Thus, $P' = T_1, t_1, \ldots, t_d$ is an SPR path between $T_1$ and $T_2$ that is permitted by $F$ and optimal with respect to $F^*$, and the claim holds.

Now, assume that $F$ does not permit $t_1$.
Because $t_1$ is one SPR away from $T_1$, it is possible to add a socket $s$ to $F$ to obtain a forest $F'$ that permits $t_1$, is also underlain by $F^*$, and is optimal.
Then, by the inductive hypothesis, there is an SPR path $P = t_1, t_2, \ldots, t_d$ that is permitted by $F'$ and, hence, $F^*$.
Let $M = m_1, m_2, \ldots, m_d$ be the sequence of moves that induce $T_1, t_1, t_2, \ldots, t_d$.
We modify $M$ to obtain $M' = m'_1, m'_2, \ldots, m'_d$ by modifying each move $m_i$ that used $s$ to instead use a different socket $s'$ in the same component of $F^*$ (each such component must have another socket because $s$ was added to a component that already had a socket in it by definition of a socket forest).
Now, the fact that $T_2 = t_d$ is permitted by $F$ implies that $s$ is not connected after applying $M$ to $T_1$ to obtain $t_d$, that is each modified move is a non terminal move of $M$ and $M'$.
Let $P' = t'_0, t'_1, \ldots, t'_d$ be the SPR path induced by $M'$.
By Corollary~\ref{cor:modify-nonterminal}, this is a valid SPR path resulting in $t'_d = t_d$, which, along with the fact that $P'$ only uses sockets of $F$, proves the claim.

\end{proof}

\reducetrees*
\begin{proof}
	We first prove that interleaving the reduction rules results in a pair of trees with the same SPR distance.
	Allen and Steel~\cite{allen01} proved that repeated application of these rules to a pair of trees $T_1$ and $T_2$ results in a pair of trees $T_1'$ and $T_2'$ on the same leaf set such that neither rule is applicable.
	Allen and Steel~\cite{allen01} also proved that the subtree reduction preserves the SPR distance.
	In combination with Theorem~\ref{thm:chain-reduction}, this proves that $\dspr{T_1, T_2} = \dspr{T_1', T_2'}$.

	Now consider the size of the reduced trees $T_1'$ and $T_2'$.
	Allen and Steel~\cite{allen01} proved that these trees have at most $28\dtbr{T_1, T_2}$ leaves.
	The fact that $\dtbr{T_1, T_2} \le \dspr{T_1, T_2}$ proves that they also have at most $28\dspr{T_1, T_2}$ leaves.
\end{proof}

\pagebreak

\bibliographystyle{spmpsci}
\bibliography{uspr}

\end{document}